%% file: main.tex
\documentclass[envcountsame,envcountsect]{llncs}
\usepackage[utf8]{inputenc}
\usepackage{amsmath}
\usepackage{stmaryrd}
\usepackage{graphicx}
\usepackage{multirow}
\usepackage{subfig}
\usepackage{tabularx, booktabs}

\usepackage{microtype}

\usepackage{soul}
\usepackage{xcolor}
\usepackage{cite}
\usepackage[noend]{algorithm2e}
\usepackage{amssymb}
\newcommand{\Label}{\mathcal L}
\newcommand{\act}{\mathit{act}}

\newcommand{\sem}[1]{\llbracket #1 \rrbracket}
\newcommand{\Inv}{\mathit{Inv}}
\newcommand{\action}{\xrightarrow{\act}}
\newcommand{\delay}[1]{\xrightarrow{#1}}
\newcommand{\delayd}{\xrightarrow{\delta}}

\makeatletter
\newcommand{\xRightarrow}[2][]{\ext@arrow 0359\Rightarrowfill@{#1}{#2}}
\makeatother

\newcommand{\cactiond}[1]{\xrightarrow{#1}_d}
\newcommand{\saction}[1]{\xRightarrow{#1}}
\newcommand{\sactiona}[1]{\xRightarrow{#1}_{\alpha}}

\newcommand{\oarr}{\mathord{\rightarrow}} %behaves like normal math symbol, not like relational symbol
\newcommand{\Ap}{\mathit{Ap}}	% atomic propositions
\spnewtheorem*{proofidea}{Proof idea}{\itshape}{\rmfamily}

\begin{document}
\title{ LTL Parameter Synthesis of Parametric Timed Automata }
\author{ Peter Bezděk \and
Nikola Beneš \and
Jiří Barnat \and
Ivana Černá
}
\institute{ Faculty of Informatics, Masaryk University, Brno, Czech Republic\\
	\email{bezdek@mail.muni.cz,\{xbenes3,barnat,cerna\}@fi.muni.cz} }
\maketitle

\begin{abstract}
  The parameter synthesis problem for parametric timed automata 
  is undecidable in general even for very 
  simple reachability properties. In this paper we introduce restrictions on 
  parameter valuations under which the parameter synthesis
  problem is decidable for LTL properties. The investigated bounded integer parameter synthesis problem could be solved using 
  an explicit enumeration of all possible parameter valuations. 
  We propose an alternative symbolic zone-based method for this problem which results in a~faster computation.
  Our technique extends the ideas of the 
  automata-based approach to LTL model checking of timed automata.
  To justify the usefulness of our approach, we provide experimental evaluation and compare our method with explicit enumeration technique.
   
\end{abstract}

\section{Introduction} 
% Intro model checking
Model checking~\cite{MC99} is a~formal verification technique applied to check
for logical correctness of discrete distributed systems. While it is often used
to prove the unreachability of a~bad state (such as an assertion violation in a
piece of code), with a~proper specification formalism, such as the \emph{Linear
  Temporal Logic} (LTL), it can also check for many interesting liveness
properties of systems, such as repeated guaranteed response, eventual stability,
live-lock, etc.

% Intro timed automata
Timed automata have been introduced in~\cite{alur94theory} and have emerged as a~useful formalism for modelling time-critical systems as found in many embedded and cyber-physical systems. 
The formalism is built on top of the standard finite automata enriched with a~set of real-time clocks and allowing the system actions to be guarded with respect to the clock valuations. In the general case, such a~timed system exhibits infinite-state semantics (the clock domains are continuous).
Nevertheless, when the guards are limited to comparing clock values with integers only, there exists a~bisimilar finite state representation of the original infinite-state real-time system referred to as the region abstraction. 
A practically efficient abstraction of the infinite-state space came with the so called zones~\cite{daws1998model}. 
The zone-based abstraction is much coarser and the number of zones \emph{reachable} from the initial state is significantly smaller. 
This in turns allows for an~efficient implementation of verification tools for timed automata, see e.g.~UPPAAL~\cite{behrmann2006uppaal}.

% The region abstraction builds on top of the observation that concrete real-time clock valuations that are between two consecutive integers are indistinguishable with respect to the valuation of an action guard. 
% Unfortunately, the size of the region-based abstraction grows exponentially with the number of clocks and the largest integer number used. 
% As a~result, the region-based abstraction is difficult to be used in practice for the analysis of more than academic toy examples, even though it has its theoretical value.

% Intro parameters
Very often the correctness of a~time-critical system relates to a~proper timing,
i.e.~it does not only depend on the logical result of the computation, but also
on the time at which the results are produced. To that end the designers are not
only in the need of tools to verify correctness once the system is fully
designed, but also in the need of tools that would help them derive proper
time parameters of individual system actions that would make the system as a
whole satisfy the required specification. After all this problem of
\emph{parameter synthesis} is more urgent in practice than the verification as
such.

% Related work undecidable parameters
\subsubsection{Related Work.}
The problem of the existence of a~parameter valuation for a~reachability property of a~parametric timed automaton in continuous time has been shown to be undecidable 
in~\cite{alur1993parametric,miller2000decidability} for a~parametric timed automaton with as few as 3~clocks. This problem remains undecidable even for integer-valued parameters~\cite{benes2015}.
A~solution for the parameter synthesis problem and reachability properties is presented in~\cite{hune02linear} where the authors provide a semi-decision algorithm which is not guaranteed to terminate in all cases. Authors also introduce a~subclass of parametric timed automata, called L/U automata for which the emptiness problem is decidable.
Decidability results for the class of L/U automata are further extended in~\cite{bozzelli2009decision}. In particular, the authors show that emptiness, finiteness and universalitity problems of the set of parameter valuations for which there is an infinite accepting run are decidable. 

% Related work bounded parameters
To obtain a~decidable version of parameter synthesis problem for parametric timed automata we need to restrict parameter valuations to bounded integers.
When modelling a~real-time system, designers can usually provide practical bounds on time parameters of individual system actions.
Therefore, introducing a~parameter synthesis method with such a~restriction is still reasonable. 
In~\cite{jovanovic2015integer} the authors show that the problem of existence of bounded integer parameter value such that a given property is satisfied is PSPACE-complete for a significant number of properties, which include Timed Computational Tree Logic. They give symbolic algorithms only for reachability and unavoidability properties.

% Based on numerical solver of ordinary differential equations Breach
% toolset~\cite{donze10breach} is used to detect sensitivity to parameters
% variation, approximate reachability analysis and parameter synthesis as well as
% robust monitoring of metric temporal interval logic.
\subsubsection{Contribution.}
The main contribution of this paper is a symbolic method that solves the
parameter synthesis problem for specifications given in the Linear Time Logic
(LTL) and parametric timed automata with bounded integer parameters. 
To this end, we introduce a~finite abstraction of parametric timed automata with bounded integer parameters and provide an~algorithm working over this abstraction. 
To evaluate our technique we implemented both a symbolic approach and explicit enumeration technique in a~proof-of-concept tool and compare the techniques on a case study. 
The finite abstraction does not provide a~unique representation of states and therefore we design an efficient state storage mechanism that deals with this problem.
The experiments demonstrate the strength of the symbolic approach which may be faster by an~order of magnitude.
%We present results confirming that it is potentially more efficient method than explicit enumeration technique. 
%We show how to apply the standard automata-based approach to LTL model checking of Vardi and Wolper~\cite{vardi86automata} in the context of an~LTL formula, a~parametric timed automaton and bounds on parameters. In particular, we show how to construct a~B\"uchi automaton coming from the parametric system under verification using a~zone-based abstraction and an~extrapolation. 

\subsubsection{Outline.}
 The rest of the paper is organised as follows. The problem definition is given
 in Section~\ref{section:preliminaries} that also introduces the basic notions.
 We then define the symbolic semantics of a parametric timed B\"{u}chi automaton and its finite abstraction in Section~\ref{section:symbolic}. 
 Section~\ref{section:algorithm} describes the parameter synthesis algorithm itself. 
Section~\ref{section:implementation} describes the implementation and used heuristics.
 Then, in Section~\ref{section:evaluation} we experimentally evaluate the proposed algorithm and compare it with explicit enumeration. 
 Finally, Section~\ref{section:conclusion} concludes the paper.

\section{Preliminaries and Problem Statement}
\label{section:preliminaries}
% parameters
% clocks
% guards
% PTA, concrete semantics
% LTL syntax, semantics
% problem statement
% solution plot
%  PTA + LTL = PTBA 
%  PTBA + zone semantics +abstraction = BA 
%  BA + cummulativeNDFS = solution
% BA
% PTBA, concrete semantics

In order to state our main problem formally, we need to describe the notion of a~parametric timed automaton. We start by describing some basic notation.

Let $P$ be a~finite set of \emph{parameters}.
An \emph{affine expression} is an expression of the form $z_0 + z_1p_1 + \ldots
+ z_np_n$, where $p_1,\ldots,p_n \in P$ and $z_0,\ldots, z_n \in \mathbb{Z}$.
We use $E(P)$ to denote the set of all affine expressions over $P$.
A~\emph{parameter valuation} is a~function $v: P \rightarrow \mathbb{Z}$ which
assigns an~integer number to each parameter.
Let $lb: P \rightarrow \mathbb{Z}$ be a~lower bound function and $ub: P \rightarrow \mathbb{Z}$ be an~upper bound function.
For an affine expression $e$, we use $e[v]$ to denote the integer value obtained by replacing each $p$ in $e$ by $v(p)$.  
We use $max_{lb,ub}(e)$ to denote the~maximal value obtained by replacing each $p$ with a~positive coefficient in $e$ by $ub(p)$ and replacing each $p$ with a~negative coefficient in $e$ by $lb(p)$. We say that the~parameter valuation $v$ respects $lb$ and $ub$ if for each $p \in P$ it holds that $lb(p) \leq v(p) \leq ub(p)$. We denote the set of all parameter valuations respecting $lb$ and $ub$ by $Val_{lb,ub}(P)$. In the following, we only consider parameter valuations from $Val_{lb,ub}(P)$.

Let $X$ be a~finite set of \emph{clocks}. 
We assume the existence of a~special \emph{zero clock}, denoted by $x_0$, 
that has always the value 0. 
A~\emph{guard} is a~finite conjunction of expressions of the form 
$x_i - x_j \sim e$ where $x_i, x_j \in X$, $e \in E(P)$ and 
$\sim\ \in \{\le, <\}$. 
We use $G(X,P)$ to denote the set of all guards over a~set of clocks $X$ 
and a~set of parameters $P$. 
A~\emph{simple guard} is a~guard containing only expressions of the form $x_i - x_j \sim e$ where $x_i, x_j \in X$, $e \in E(P)$,  
$\sim\ \in \{\le, <\}$, and $x_i=x_0$ or $x_j=x_0$. We also use $\overline{G}(X,P)$ to denote the set of all simple guards over a~set of clocks $X$ 
and a~set of parameters $P$. 
A~\emph{clock valuation} is a~function $\eta: X \rightarrow \mathbb{R}_{\geq0}$
assigning non-negative real numbers to each clock such that $\eta(x_0)=0$. 
We~denote the set of all clock valuations by $Val(X)$.
Let $g \in G(X,P)$ and $v$ be a~parameter valuation and $\eta$ be a~clock valuation.
Then $g[v,\eta]$ denotes a~boolean value obtained from~$g$ by replacing each 
parameter $p$ with $v(p)$ and each clock $x$ with $\eta(x)$. 
A~pair $(v,\eta)$ \emph{satisfies} a~guard $g$, denoted by $(v,\eta)\models g$,
if $g[v,\eta]$ evaluates to true. 
The~\emph{semantics} of a~guard $g$, denoted by $\sem{g}$, 
is a~set of all valuation pairs $(v,\eta)$ such that $(v,\eta) \models g$.
For a~given parameter valuation $v$ we write $\sem{g}_v$ 
for the set of clock valuations $\{\eta\ |\ (v,\eta)\models g \}$.

We define two operations on clock valuations.
Let $\eta$ be a~clock valuation, $d$~a~non-negative real number and 
$R \subseteq X$ a~set of clocks.
We use $\eta + d$ to denote the clock valuation that adds the delay $d$ 
to each clock, i.e.~$(\eta + d)(x) = \eta(x) + d$ for all $x \in X \setminus \{x_0\}$.
We further use $\eta\langle R\rangle$ to denote the clock valuation that resets clocks 
from the set $R$, i.e.~$\eta\langle R\rangle(x) = 0$ if $x \in R$, $\eta\langle R\rangle(x) = \eta(x)$ otherwise.

%We can now proceed with the definition of a~parametric timed automaton and its semantics.

\begin{definition}[PTA]
A~\emph{parametric timed automaton} (PTA) is a~tuple 
$M = (L, l_0, X, P, \Delta, \Inv)$ where
\begin{itemize}
  \item $L$ is a~finite set of locations,
  \item $l_0 \in L$ is the initial location, 
  \item $X$ is a~finite set of clocks,
  \item $P$ is a~finite set of parameters,
  \item $\Delta \subseteq L \times \overline{G}(X,P) \times 2^X \times L$ is 
	a~finite transition relation, and
  \item $\Inv: L \rightarrow \overline{G}(X,P)$ is an invariant function.
\end{itemize}
\end{definition}
We use $q \xrightarrow{g,R}_\Delta q'$ to denote $(q,g,R,q') \in \Delta$. The semantics of a~PTA is given as a labelled transition system. 
A~\emph{labelled transition system} (LTS) over a~set of symbols $\Sigma$
is a~triple $(S, s_0, \oarr)$, where $S$ is a~set of states, $s_0 \in S$ is an initial state and $\oarr \subseteq S \times \Sigma \times S$ is a~transition relation. 
We use $s\xrightarrow{a} s'$ to denote $(s,a,s') \in \oarr$. 

\begin{definition}[PTA semantics]
Let $M = (L, l_0, X, P, \Delta, \Inv)$ be a~PTA and $v$ be a~parameter valuation. 
The~\emph{semantics} of $M$ under $v$, denoted by 
$\sem{M}_v$, is an~LTS $(\mathbb{S}_M,s_0,\oarr)$ over the set of symbols
$\{ \act \} \cup \mathbb R_{\ge 0}$,
 where
\begin{itemize}
\item $\mathbb{S}_M=L\times Val(X)$ is a~set of all states,
\item $s_0 = (l_0,\mathbf{0})$, where $\mathbf{0}$ is a~clock valuation
with $\mathbf{0}(x) = 0$ for all $x$, and
\item the transition relation $\oarr$ is specified
for all $(l,\eta), (l',\eta') \in \mathbb S_M$ 
as follows:
 \begin{itemize}
 \item $(l,\eta) \delay{d} (l',\eta')$ if $l=l'$, 
	$d \in \mathbb R_{\ge 0}$, $\eta' = \eta + d$, and $(v,\eta') \models \Inv(l')$,
 \item $(l,\eta) \action (l',\eta')$ if 
	$\exists g, R : l \xrightarrow{g,R}_\Delta l'$, 
	$(v,\eta) \models g$, $\eta' = \eta\langle R\rangle$, \\ and $(v,\eta') \models \Inv(l')$.
 \end{itemize}
 The transitions of the first kind are called \emph{delay transitions},
 the latter are called \emph{action transitions}.
\end{itemize}

We write $s_1 \cactiond{act} s_2$ if there exists $s' \in \mathbb{S}_M$ and $d \in \mathbb{R}^{\geq 0}$ such that $s_1 \stackrel{act}{\longrightarrow} s' \stackrel{d}{\longrightarrow} s_2$.
A~proper run $\pi$ of $\sem{M}_v$ is an infinite alternating sequence of delay and action transitions that begins with a~delay transition $\pi = (l_0, \eta_0) \delay{d_0} (l_0,\eta_0+d_0) \action (l_1,\eta_1) \delay{d_1}
\cdots$. A~proper run is called Zeno if the sum of all its delays is finite.
\end{definition}

%For the rest of the paper, we assume that we only deal with a~deadlock-free PTA, i.e.~that for each considered parameter valuation $v$ there is no state without a~reachable action transition in $\sem{A}_v$. We deal with Zeno runs later.

Let $M$ be a~PTA, $\Label : L \to 2^\Ap$ be a~labelling function that assigns a~set of atomic propositions to each location of $M$, $v$ be a~parameter valuation, and $\varphi$ be an~LTL formula. 
We say that $M$ under $v$ with~$\Label$ satisfies $\varphi$, denoted by $(M, v, \Label) \models \varphi$ if for all proper runs $\pi$ of $\sem{M}_v$, $\pi$ satisfies $\varphi$ where atomic propositions are determined by $\Label$.

Given a~parametric timed automaton $M$, a~labelling function $\mathcal L$, and an~LTL property $\varphi$, \textit{the parameter synthesis problem} is to compute the set of all parameter valuations $v$ such that $(M,v,\Label) \models \varphi$. 
Unfortunately, it is known that the parameter synthesis problem for a~PTA is undecidable even for very simple (reachability) properties~\cite{alur1993parametric}. 
Instead of solving the general problem, we thus focus on a~more constrained version which is still reasonable for practical purposes.

\subsubsection{Problem Formulation.}
%\begin{problem}[The bounded integer parameter synthesis problem]
%\label{problem:bounded}
Given a~parametric timed automaton $M$, a~labelling function $\mathcal L$,
an~LTL property $\varphi$, a~lower bound function $lb$ and an~upper bound function $ub$, 
\textit{the bounded integer parameter synthesis problem} is to compute the set of all parameter valuations
$v$ such that $(M,v,\Label) \models \varphi$ and $lb(p) \le v(p) \le ub(p)$.
%\end{problem}

This problem is trivially decidable using the standard zone-based abstraction
and explicit enumeration of all parameter valuations. In order to avoid the
necessity of the explicit enumeration of all parameter valuations we use
a~combination of the zone-based abstraction and a~symbolic representation of
parameter valuation sets. Our algorithmic framework which solves this problem
consists of three steps.

As the first step, we apply the~standard automata-based LTL model checking of timed automata~\cite{alur94theory} to parametric timed automata. We employ this approach in the following way. 
From a~PTA $M$ and an~LTL formula $\varphi$ we produce a~product parametric timed B\"{u}chi automaton (PTBA) $A$. The accepting runs of the automaton $A$ correspond to the runs of $M$ violating the formula $\varphi$.% (analogously as in the case of timed automata).

As the second step, we employ a~symbolic semantics of a~PTBA $A$ with a~suitable extrapolation. From the symbolic state space of a~PTBA $A$ we finally produce a~B\"{u}chi automaton $B$ in which each state is associated symbolic information about parameter valuations. This transformation is described in Section~\ref{section:symbolic}.

As the last step, we need to detect all parameter valuations for which there
exists an~accepting run in B\"{u}chi automaton $B$.
%This is done using our Cumulative NDFS algorithm. 
To that end, we employ a~new algorithm, which we call the Cumulative NDFS. The
algorithm is described in detail in Section~\ref{section:algorithm}.

We now proceed with the definitions of a~B\"{u}chi automaton, a~parametric timed B\"{u}chi automaton and its semantics. 

\begin{definition}[BA]
A~\emph{B\"{u}chi automaton} (BA) is a~tuple $B=(Q,q_0,\Sigma,\mathord{\rightarrow}, F)$, where $Q$ is a~finite set of states, $q_0 \in Q$ is an~initial state, $\Sigma$ is a~finite set of symbols, $\rightarrow \subseteq Q \times \Sigma \times Q$ is a~set of transitions, and $F \subseteq Q$ is a~set of accepting states (acceptance condition).
An~$\omega$-word $w=a_0a_1a_2\ldots \in \Sigma^{\omega}$ is accepting if there is an infinite sequence of states $q_0q_1q_2\ldots$ such that $q_i\stackrel{a_i}{\longrightarrow}q_{i+1}$ for all $i \in \mathbb{N}$, and there exist infinitely many $i\in\mathbb{N}$ such that $q_i \in F$.
\end{definition}
 
\begin{definition}[PTBA]
A~\emph{parametric timed B\"{u}chi automaton} (PTBA) is a~pair $A = (M, F)$ where
 $M=(L, l_0, X, P, \Delta, \Inv)$ is a~PTA, and
 $F \subseteq L$ is a~set of accepting locations.
\end{definition}

Zeno runs represent non-realistic behaviours and it is desirable to ignore them in
analysis. Therefore, we are interested only in non-Zeno accepting runs of
a~PTBA. There is a~syntactic transformation to the so-called strongly non-Zeno
form~\cite{tripakis2005checking} of a~PTBA, which guarantees that each
accepting run is non-Zeno. For the rest of the paper,
%we assume that we have the strongly
%non-Zeno form of a~PTBA, as introduced in~\cite{tripakis2005checking}.
we thus assume that there are no Zeno accepting runs in the PTBA.

\begin{definition}[PTBA semantics]
Let $A = (M, F)$ be a~PTBA and $v$ be a~parameter valuation. 
The~\emph{semantics} of $A$ under $v$, denoted by 
$\sem{A}_v$, is defined as $\sem{M}_v=(\mathbb{S}_M,s_0,\oarr)$.

We say a~state $s=(l,\eta) \in \mathbb{S}_M$ is \emph{accepting} if $l \in F$.
A~proper run $\pi=s_0 \stackrel{d_0}{\longrightarrow} s_0' \stackrel{act}{\longrightarrow} s_1 \stackrel{d_1}{\longrightarrow} s_1' \stackrel{act}{\longrightarrow} \ldots$ of $\sem{A}_v$ is accepting if there exists an~infinite set of indices $i$ such that $s_i$ is accepting. 
\end{definition}

\section{Symbolic Semantics}
\label{section:symbolic}
In this section we show the construction of a~finite system which represents
the semantics of a~given PTBA. First, we describe a parametric extension of the
zone abstraction. This extension is based on %an usage of
constrained parametric difference bound matrices, described in~\cite{hune02linear}.
However, this abstraction itself does not guarantee finiteness in our setting.
To solve this problem we further introduce a finite parametric extrapolation.

\subsection{Constrained Parametric Difference Bound Matrix}
A~\textit{constraint} is an inequality of the form $e \sim e'$ where $e, e'  \in E$  and  $\sim\ \in \{ \mathord{>}, \mathord{\ge},$ $\mathord{\le}, \mathord{<} \}$. 
We define $c[v]$ as the boolean value obtained by replacing each $p$ in $c$ by $v(p)$. 
A~valuation $v$ \textit{satisfies} a~constraint $c$, denoted $v \models c$, if $c[v]$ evaluates to true. 
The \textit{semantics} of a~constraint $c$, denoted $\sem{c}$, is the set of all valuations that satisfy $c$.
A~finite set of constraints $C$ is called a~\textit{constraint set}. 
A~valuation \textit{satisfies} a~constraint set $C$ if it satisfies each $c \in C$.
The \textit{semantics} of a~constraint set $C$ is given by $\sem{C}  = \bigcap_{c \in C} \sem{c}$.
A~constraint set $C$ is \textit{satisfiable} if $\sem{C} \neq \emptyset$.
A~constraint $c\ \mathit{ covers}$ a~constraint set $C$, denoted $C \models c$, if $\sem{C} \subseteq \sem{c}$. 
 
As in~\cite{hune02linear}, we identify the~relation symbol $\leq$ with the
boolean value true and $<$ with the boolean value false. Then, we treat boolean
connectives on relation symbols $\leq$, $<$ as operations with boolean values. For
example, $(\leq \implies <) = \mathord{<}$.

We now define the parametric difference bound matrix, the constrained
parametric difference bound matrix, several operations on them, and the
symbolic semantics of a~PTBA.
%The concept of a~constrained parametric difference bound matrix follows the idea used %in~\cite{hune02linear}.

\begin{definition} 
A \emph{parametric difference bound matrix} (PDBM) over $P$ and $X$ is a~set $D$ which contains for all $0 \leq i, j \leq |X|$ a~guard of the form $x_i - x_j \prec_{ij} e_{ij}$ where $x_i,x_j \in X$ and $ e_{ij} \in E(P)\cup \{\infty\}$ and $ i=j \implies e_{ii}=0$.
We denote by $D_{ij}$ a~guard of the form $x_i - x_j \prec_{ij} e_{ij}$ contained in $D$ . 
Given a~parameter valuation $v$, the~\emph{semantics} of $D$ is given by $\sem{D}_{v} =\sem{ \bigwedge_{i,j} D_{ij} }_{v}$. 
A PDBM $D$ is \emph{satisfiable} with respect to $v$ if $\sem{D}_v$ is non-empty.
%If $f$ is a~guard of the form $x_i -x_j \prec e$ with $i \neq j$ (a~proper guard), then $D[f]$ denotes the PDBM obtained from $D$ by replacing $D_{ij}$ with $f$. 
%We denote by $\mathit{PDBMS}(P,X)$ the set of all PDBM over parameters $P$ and clocks $X$.
\end{definition}

\begin{definition}
A \emph{constrained parametric difference bound matrix} (CPDBM) is a~pair $(C,D)$, where $C$ is a~constraint set and $D$ is a~PDBM and for each $0\leq i\leq |X|$ it holds that $C\models e_{0i}\geq 0$. 
The \emph{semantics} of $(C,D)$ is given by $\sem{C,D} = \{ (v,\eta)\ |\ v \in \sem{C} \wedge \eta \in \sem{D}_{v} \}$.
We call $(C,D)$ \emph{satisfiable} if $\sem{C,D} $ is non-empty. 
%We denote by $\mathit{CPDBMS}$ the set of all CPDBM.
A CPDBM $(C,D)$ is said to be \emph{in the canonical form} if and only if for all $i,j,k$, $C \models  e_{ij} (\prec_{ik} \wedge \prec_{kj}) e_{ik}+e_{kj}$.
%Note that the existence of a canonical form is not guaranteed.
\end{definition}
%The transition relation between states is defined with the help of the following three operations.

\subsubsection{Resetting a~Clock.}
Suppose $(C,D)$ is a~CPDBM in the canonical form. 
The reset of the clock $x_r$ in $(C,D)$, denoted by $(C,D)\langle x_r \rangle$,
is given as $(C,D\langle x_r\rangle)$ where: %$D\langle x_r \rangle$ is defined as follows:
\[
 D\langle x_r \rangle_{ij}=
  \begin{cases}
   D_{0j}      	& \text{if } i \not= j \text{ and } i = r\text{,}	 \\
   D_{i0}       	& \text{if } i \not= j \text{ and } j = r\text{,}      \\
   D_{ij}  	& \text{else.}
  \end{cases}
\]
We can again generalise this definition to a~set of clocks:\\
$ (C,D)\langle x_{i_0},x_{i_1},\ldots, x_{i_k} \rangle \stackrel{def}{\Leftrightarrow} (C,D)\langle x_{i_0}\rangle\langle x_{i_1}\rangle \ldots\langle x_{i_k}\rangle$.

\subsubsection{Applying a~Guard.}
Suppose g is a~guard of the form $x_i -x_j \prec e$, $(C,D)$ is a~CPDBM in
the canonical form and $D_{ij} = (e_{ij}, \prec_{ij})$.
The application of the guard $g$ on $(C,D)$ generally results in a~set of
CPDBMs and is defined as follows:
\[
 (C, D)[g ]=
  \begin{cases}
   \{ (C, D[g]) \}      				
	& \text{if }  C \models \neg (e_{ij} ( \prec_{ij} \implies \prec ) e)\text{,}\\  %no
   \{ (C, D) \}    				
	& \text{if }  C \models  e_{ij} ( \prec_{ij} \implies \prec ) e\text{,}\\ %yes
   \{(C\cup\{ e_{ij} ( \prec_{ij} \implies \prec ) e \}, D),  & \text{otherwise,} \\ 
       (C\cup\{ \neg e_{ij} ( \prec_{ij} \implies \prec ) e \}, D[g]),\}	
		&			\\%split
  \end{cases}
\]
where $D[g]$ is defined as follows:
\[
 D[g]_{kl}=
  \begin{cases}
   ( e, \prec)      	& \text{if } k = i \text { and } l = j\text{,}  \\
   D_{kl}  	& \text{else.}
  \end{cases}
\]
We can generalise this definition to conjunctions of guards as follows:\\ 
$ D[g_{i_0} \wedge g_{i_1} \wedge \ldots \wedge g_{i_k}] \stackrel{def}{\Leftrightarrow} D[g_{i_0}][g_{i_1}]\ldots[g_{i_k}]$.

\subsubsection{Time Successors.}
Suppose $(C,D)$ is a~CPDBM in the canonical form.
The time successor of $(C,D)$, denoted by $(C,D)^\uparrow$,
 represents a~CPDBM with all upper bounds on clocks removed and is 
 given as $(C,D^\uparrow)$ where:
\[
 D^\uparrow_{ij} =
  \begin{cases}
   (\infty, <) & \text{if } i \not= 0 \text{ and } j = 0\text{,}	 \\
   D_{ij}        & \text{else.} 
  \end{cases}
\]

The~reset and time successor operations preserve the canonical form of a~CPDBM.
After the application of a~guard the CPDBM may no longer be in the canonical form
and thus a transformation to the~canonical form needs to be performed.
However, due to the presence of parameters the standard canonisation~\cite{dill1990timing} process can be ambiguous.
%Therefore, we use extended canonisation 
The canonisation procedure is therefore extended 
to cope with this ambiguity.
As a consequence, the result of the canonisation is not a~single CPDBM, but may
generally be a~set containing potentially more CPDBMs in the canonical form
with mutually disjoint constraint sets.

To canonise the~given CPDBM we need to derive the tightest constraint on each
clock difference. Deriving the tightest constraint on a~clock difference can be
seen as finding the shortest path in the~graph interpretation of the
CPDBM. In~\cite{hune02linear} the authors implement the canonisation using
a~nondeterministic extension of the Floyd-Warshall algorithm where on each relaxation
a~split into two different CPDBMs can occur.

\subsubsection{Canonisation.}
First, we define a~relation $\longrightarrow_{FW}$ on constrained parametric
bound matrices as follows, for all $0\leq k,i,j \leq |X|$:

   \begin{itemize}
   \item 
   $(k,i,j,C_1,D_1) \longrightarrow_{FW} (k,i,j+1,C_2,D_2)$ \\
   if $(C_2,D_2) \in (C_1,D_1)[x_i-x_j (\prec_{ik} \wedge \prec_{kj} ) e_{ik}+e_{kj} ]$
   \item
   $(k,i,|X|+1,C_1,D_1) \longrightarrow_{FW} (k,i+1,0,C_1,D_1)$
%   if $(C_2,D_2) \in (C_1,D_1)[x_i-x_j (\prec_{ik} \wedge \prec_{kj} ) e_{ik}+e_{kj} ]$
   \item
   $(k,|X|+1,0,C_1,D_1) \longrightarrow_{FW} (k+1,0,0,C_1,D_1)$
%   if $(C_2,D_2) \in (C_1,D_1)[x_i-x_j (\prec_{ik} \wedge \prec_{kj} ) e_{ik}+e_{kj} ]$
   
   \end{itemize}

The~relation $\longrightarrow_{FW}$ can be seen as a~representation of the computation steps of the extended  Floyd-Warshall algorithm. 

Suppose now $(C,D)$ is a~CPDBM. The~canonical set of $(C,D)$, denoted as
$(C,D)_c$, represents a~set of CPDBMs with the tightest constraint on each
clock difference in $D$ and is defined as follows:
\[
 (C, D)_c =
   \{ (C', D') \mid (0,0,0,C,D) \longrightarrow_{FW}^* (|X|+1,0,0, C',D') \}      \]

\begin{example}\label{ex:canon}
Let $x,y \in X$ and $p,q \in P$. 
For a CPDBM $(C,D)=(\emptyset, \{x\leq p, y \leq q, y\leq x, y \leq x\})$ we obtain by canonisation $(C,D)_c = \{ (\{p \leq q\},\{ x\leq p, y \leq p, y\leq x, y \leq x\})$ , $(\{q < p\},\{ x\leq q, y \leq q, y\leq x, y \leq x\}) \}$.
\end{example}

\begin{definition}[PTBA symbolic semantics]
Let $A = ((L, l_0, X, P, \Delta, \Inv),$ $F)$ be a~PTBA. Let $lb$ and $ub$ be a~lower bound function and an~upper bound function on parameters.
The \emph{symbolic semantics} of $A$ with respect to $lb$ and $ub$ is a~transition system $(\mathbb S_A,\mathbb S_{init},\Longrightarrow)$, denoted as $\sem{A}_{lb,ub}$, where 
\begin{itemize}
\item $\mathbb S_A = L \times \{ \sem{C,D} \mid (C,D) $ is a $CPDBM\}$ is the set of all symbolic states,
\item the set of initial states $\mathbb S_{init} =
  %$ is defined as $
  \{(l_0, \sem{C, D})\ |\  (C,D) \in (\emptyset, E^\uparrow)[ \mathit{Inv}(l_0)] \}$, where
	\begin{itemize}
		\item $E$ is a~PDBM with $E_{i,j}=(0,\leq)$ for each $i,j$, and
		\item for each $p \in P$, the constraints $p \geq lb(p)$ and $p \leq ub(p)$ are in $C$.
 	\end{itemize}
\item There is a~transition $(l,\sem{C,D}) \Longrightarrow (l',\sem{C'_c,D'_c})$ if  
\begin{itemize}
\item $l \stackrel{g,R}{\longrightarrow_{\Delta}} l'$ and
\item $(C'', D'') \in (C,D)[g]$ and
\item $(C''_c, D''_c) \in (C'', D'')_c$ and
\item $ (C', D') \in  (C''_c,D''_c\langle R\rangle^\uparrow)[Inv(l')]$ and 
\item $(C'_c, D'_c) \in (C', D')_c$.
\end{itemize}
\end{itemize} 

%A symbolic state is represented by a~tuple $(l,\sem{C,D})$ where $l$ is a~location, $(C,D)$ is a~CPDBM.  
We say that a~state $S=(l,\sem{C,D}) \in \mathbb S_A$ is accepting if $l \in F$. 
We say that $\pi = S_0 \Longrightarrow S_1 \Longrightarrow \ldots$ is a~run of
$\sem{A}_{lb,ub}$ if $S_0 \in \mathbb S_{init}$ and for each $i$, $S_i \in S_A$ and
$S_{i-1}\Longrightarrow S_i$. A~run respects a~parameter valuation $v$ if for
each state $S_i=(l_i,\sem{C_i,D_i})$ it holds that $v \in \sem{C_i}$. 
A run $\pi$ is accepting if there exists an infinite set of indices $i$ such that $S_i$ is accepting. For the rest of the paper we fix $lb$, $ub$ and use $\sem{A}$ to denote $\sem{A}_{lb,ub}$.
\end{definition}

\subsection{Finite Abstraction}
Similarly to the nonparametric case, the symbolic
transition system $\sem{A}$ may be infinite. In order to obtain a~finite
transition system we need to apply a~finite abstraction over $\sem{A}$. In the
standard case of timed automata without parameters we use one of the
extrapolation techniques~\cite{bouyer2004forward,behrmann2006lower}. In our parametric setup we define a~new finite
abstraction called the \emph{pk-extrapolation} which is a~parametric extension
of the widely used \emph{k-extrapolation}~\cite{bouyer2004forward}.
The k-extrapolation identifies states which are identical except for the clock values which exceeds the maximal constant from guards and invariants.

In our parametric setup, we need to define the maximal constant with which each clock within a~PTBA is compared. 
We define $M(x)$ as the maximal value in $\{max_{lb,ub}(e) \mid e$~is compared
with $x$ in a~guard or an~invariant of the considered PTBA$\}$.
The core idea of pk-extrapolation is the same as the idea of k-extrapolation. We substitute each bound on clock difference in the CPDBM whenever this bound exceeds the maximal constant. The precise description of this substitution process is given in the Definition \ref{definition:pkextrapolation}.
%Note that we have separate maximal constant for each clock and exact description of this substitution is given in definition \ref{definition:pkextrapolation}.}
%The pk-extrapolation performs bound substitutions in the same way.
Contrary to the nonparametric case, due to the occurrence of parameters in the CPDBM bounds, the substitution
process may be ambiguous.
In these situations we restrict the parameter values in order to obtain an unambiguous situation.
This solution is similar to the constraint set splitting that is done in the
application of a~guard and in the canonisation procedure.
%We have to consider all possible restrictions in such situations.
%Therefore we consider all unambiguous cases and update constraint set for each case accordingly. 
Therefore, the result of pk-extrapolation is a~set of CPDBMs instead of a~single CPDBM.
% We follow with definition and a small example.
 
\begin{definition}
\label{definition:pkextrapolation}
Let $A$ be a~PTBA, $(l,\sem{C,D})$ be a~symbolic state of $\sem{A}$ and $D_{ij} = x_i - x_j \prec_{ij} e_{ij}$ for each $0\leq i,j\leq |X|$. 
We define the \emph{pk-extrapolation} $\alpha_{pk}$ in the following way.
$\alpha_{pk}(l,\sem{C,D})$ is the set of all $(l,\sem{C',D'})$ such that for each
$i$, $j$, $0 \leq i,j \leq |X|$ one of the following conditions holds:
\begin{itemize}
\item
%it holds that
$D'_{ij} = x_i - x_j  \prec_{ij} e_{ij}$ and the constraint $(e_{ij} \leq M(x_i)) \in C'$,
\item
%it holds that
$D'_{ij} = x_i - x_j  < \infty$ and the constraint $(e_{ij} > M(x_i))\in C'$,
\item
%it holds that
$D'_{ij} = x_i - x_j  \prec_{ij} e_{ij}$ and the constraint $(e_{ij} \geq -M(x_j))\in C'$,
\item
%it holds that
$D'_{ij} = x_i - x_j  < -M(x_j)$ and the constraint $(e_{ij} < -M(x_j))\in C'$.
\end{itemize}
\end{definition}

\begin{example}\label{exa:extrapolation}
Consider $x,y \in X$, $p \in P$, $p \in [0,7]$, $M(x)=M(y)=10$, and the symbolic
state $(l,\sem{C,D})$ where $C=\emptyset$ and $D=\{x\leq y, y\leq x, y \leq
2p\}$. Now, $\alpha_{pk}(l,\sem{C,D})$ contains two symbolic states:
$(l,\sem{C_1,D_1})$ and $(l,\sem{C_2,D_2})$ where $C_1=\{2p \leq 10\}$,
$D_1=\{x\leq y, y\leq x, y \leq 2p\}$, $C_2=\{2p>10\}$, $D_2=\{x\leq y, y\leq
x, y < \infty\}$.
\end{example}

\begin{theorem}
\label{theorem:extraLassoSimulation}
Let $A$ be a PTBA. The pk-extrapolation is a~finite abstraction that preserves all accepting runs of $\sem{A}_v$ for each parameter valuation $v$.
\end{theorem}
The proof of this theorem is given in Appendix~\ref{app:extraLassoSimulation}.
\iffalse
\begin{proofidea}
We can transform the proof of Theorem 1 of \cite{li09formal} as well as the corresponding lemmata and definitions into our parametric setup.
Due to space constraints, we did not include the full technically detailed proof and we kindly refer the reader to~\cite{arxiv2016PtaSynthesis}.
\end{proofidea}
\fi

\section{Parameter Synthesis Algorithm}
\label{section:algorithm}
% Algorithm input, Algorithm goal, Assumptions
We recall that our main objective is to find all parameter valuations for which the parametric timed automaton satisfies its specification. 
In the previous sections we have described the standard automata-based method employed under a~parametric setup which produces a~B\"{u}chi automaton.
For the rest of this section  we use $s.\sem{C}$ to denote the set $\sem{C}$ where $s=(l,\sem{C,D})$ is a~state of the input B\"{u}chi automaton.
%we denote for each state $s=(l,\sem{C,D})$ of the B\"{u}chi automaton on the~input the set of valuations $\sem{C}$ as $s.\sem{C}$.
We say that a~sequence of states $s_1 \Longrightarrow s_2 \Longrightarrow \ldots \Longrightarrow s_n \Longrightarrow s_1$ is a~cycle under the~parameter valuation $v$ if each state $s_i$ in the sequence satisfies $v \in s_i.\sem{C}$. A~cycle is called accepting if there exists $0\leq i\leq n$ such that $s_i$ is accepting.

The standard automata-based LTL model checking checks the emptiness of the produced B\"{u}chi automaton. The emptiness check can be performed using the Nested Depth First Search (NDFS) algorithm~\cite{courcoubetis1992memory}.
The NDFS algorithm is a~modification of the depth first search algorithm which allows a detection of an~accepting cycle in the given B\"{u}chi automaton.

% Extension of NDFS, basic idea
Contrary to the standard LTL model checking, it is not enough to check the
emptiness of the produced B\"{u}chi automaton.
Our objective is to check the emptiness of the produced B\"{u}chi automaton for each considered parameter valuation.
To solve this objective, we introduce a~new algorithm called the~Cumulative NDFS algorithm which is an~extension of the NDFS algorithm. The pseudocode of Cumulative NDFS is given in Algorithm~1.
Our modification is based on the set $Found$ which accumulates all detected parametric valuations such that an~accepting cycle under these valuations was found. 
In contrast to the NDFS algorithm, whenever Cumulative NDFS detects an~accepting cycle, parameter valuations are saved to the set $Found$ and the computation continues with a~search for another accepting cycle.
Note the fact that whenever we reach a~state $s'$ with $s'.\sem{C} \subseteq Found$ we already have found an~accepting cycle under all valuations from $s'.\sem{C}$ and there is no need to continue with the search from $s'$.
Therefore, we are able to speed up the computation whenever we reach such a~state.

The crucial property the algorithm is based on is that of monotonicity.
The set of parameter valuations $s.\sem{C}$ can not grow along any run of the input automaton.
Lemma~\ref{lemma:mono} states this observation formally.
The observation follows from the definition of successors in $\sem{A}^{\alpha}$
and the definition of operations on CPDBMs. The clear corollary of
Lemma~\ref{lemma:mono} is the fact that each state $s$ on a~cycle has the same
set $s.\sem{C}$. 
\begin{lemma}
\label{lemma:mono}
Let $A$ be a~PTBA, $\alpha$ be an abstraction and $s$ be a~state in $\sem{A}^{\alpha}$. 
For every state $s'$ reachable from $s$ it holds that $s'.\sem{C} \subseteq s.\sem{C}$.
\end{lemma}

\begin{algorithm}[t]
  \label{algorithm:cumulative}
  \caption{Cumulative NDFS}
  \SetAlgoLined\DontPrintSemicolon
  
  \SetInd{0.65em}{0.65em}
  \SetKwProg{alg}{Algorithm}{}{}
  \SetKwProg{proc}{Procedure}{}{}

   \begin{multicols}{2}
  \setcounter{AlgoLine}{0}
  \alg{$\mathit{CumulativeNDFS}(G)$}
  {
    \nl $Found \gets \emptyset$; $Stack \gets \emptyset$\;
        $Outer \gets \emptyset$; $Inner \gets \emptyset$\;
    \nl  $OuterDFS(s_{init})$ \;
    \nl \KwRet $ Accepted \gets Found$\;
   }{}

  \proc{$OuterDFS(s)$}
  {
    \nl $Stack \gets Stack \cup \{s\}$ \;
    \nl $Outer \gets Outer \cup \{s\}$ \;
    \nl \ForEach{$s'$ such that $s \rightarrow s'$ } % \from $Succ(s)$}
    {
      \nl \If{ $s'  \notin Outer\ \wedge $ $s'  \notin Stack\ \wedge  $ $s'.\sem{C} \not\subseteq Found$}
      {
        \nl $OuterDFS(s')$ \;
      }
    } 
    \nl  \If{ $s  \in Accepting \wedge s.\sem{C} \not\subseteq Found$}
      {
        \nl $InnerDFS(s)$ \;
      }
    \nl $Stack \gets Stack \setminus \{s\}$ \;
%    \nl \KwRet\;
  }
  
  \proc{$InnerDFS(s)$}
  {
    \nl $Inner \gets Inner \cup \{s\}$ \;
    \nl \ForEach{$s'$ such that $s \rightarrow s'$ } % \from $Succ(s)$}
    {
      \nl \If{ $s'  \in Stack$}
      {
        \nl ``Cycle detected''\;
        \nl $Found \gets Found \cup s'.\sem{C}$ \;
        \nl \KwRet \;
        
      }
      \nl \If{ $s'  \notin Inner\ \wedge $ \\
      $s'.\sem{C} \not\subseteq Found $}
      {
        \nl $InnerDFS(s')$ \;
      }
    } 
%    \nl \KwRet\;
  }
  \end{multicols}
\end{algorithm}

\begin{theorem}
\label{theorem:CumulativeNDFSCorrectness}
Let $A$ be a~PTBA and $\alpha$ an~abstraction over $\sem{A}$. 
A~parameter valuation $v$ is contained in the output of the \textit{CumulativeNDFS($\sem{A}^{\alpha}$)} if and only if there exists an~accepting run respecting $v$ in $\sem{A}^{\alpha}$.
\end{theorem}
The proof of this theorem is given in Appendix~\ref{app:cumulative}.
 
\bigskip
%We recall that our objective was to synthesise the set of all parameter valuations such that the given parametric timed automaton satisfies the given LTL property. In order to compute this set we employed a~zone-based semantics, an~extrapolation technique and the Cumulative NDFS algorithm. We have shown the way to compute all parameter valuations for which the given LTL formula is not satisfied. 
As the last step in the solution to our problem, we need to complement the set $Accepted$. Thus, the solution is the complement of the set $Accepted$, more precisely the set $Val_{lb,ub}(X,P)\setminus Accepted$. To conclude this section, we state that Theorem \ref{theorem:CumulativeNDFSCorrectness} together with Theorem \ref{theorem:extraLassoSimulation} imply the correctness of our solution.% to Problem \ref{problem:bounded}.

\section{Implementation}
\label{section:implementation}
We have implemented our approach in a~proof-of-concept tool.
We are able to process models given as networks of parametric timed automata.
A~network represents a~product of several parametric timed automata where
handshake synchronization of two components at a time is allowed. We also
extend the parametric timed automata with data variables which enable the usage
of guards on data values and transition effects on data values. Such model is
considered standard in the field and is used as the modelling language in the tool
UPPAAL. 

\textbf{Deadlocks} Cumulative NDFS algorithm returns all parameter valuations for which LTL property does not hold.
However, state space can contain deadlock states which also need to be detected and reported.
In the~nonparametric setting a~state is a~deadlock state if there are no enabled outgoing transitions.
In~a~parametric setting the deadlock status of a~state depends on the parameter valuation. % and time can not evolve.
%Detection of deadlock is not as straightforward as in LTL model checking. 
To decide for which parameter valuations a~state $(l,\sem{C,D})$ is a deadlock we need to consider all guards $g_1,\ldots, g_n$ of the outgoing transitions of $l$. The state $(l,\sem{C,D})$ is a~deadlock for all parameter valuations in $\sem{C,D}[\neg g_1 \wedge \ldots \wedge \neg g_n]$. Applying this detection to each reachable state, all parameter valuations leading to deadlock are detected during computation.

\textbf{State space storage} One of the performance critical parts of the implementation is the state space
storage. We use the state space storage to look up and store information about
presence of each state in the sets $Inner$, $Outer$, and $Stack$. We refer to
this information as $data$.
A~straightforward implementation would simply store each state together with its data. 
Such a~solution is only efficient when a~unique representation of states is available. 
Without such a~unique representation the storage operations have to perform
expensive equivalence checks with each stored state in the worst case scenario. 
In~\cite{jovanovic2015integer} the authors introduce unique representation
based on a computation of an~integer hull. The integer hull of a~given set is a
convex hull of all integer elements of a~given set.

The solution of~\cite{jovanovic2015integer} assumes the existence of an~upper
bound for each clock. We do not have such an upper-bound assumption and
therefore this solution is not directly applicable in our technique.
However, we use the integer hull as a heuristic approximation of a~unique
representation of a~CPDBM instead. This way we obtain a practically efficient solution
that deals with the non-existence of a~unique representation of a~state.

The solution is based on two mappings. The first mapping, denoted by $M_1$
%Now, we introduce practically efficient solution to non-existence of a~unique representation of state. 
%The idea is to use an integer hull as an approximation of unique representation of a CPDBM. 
%We introduce mapping $M_1$ which 
maps a given integer hull to a list of CPDBM representations.
Each such list contains the representations of semantically different CPDBMs with the same integer hull. 
Thanks to $M_1$ we can quickly distinguish states with different integer hulls.
However, each storage operation still needs to perform the expensive computation of the integer hull. 
In order to reduce number of the integer hull computations, we introduce the second mapping, denoted by $M_2$.
This second mapping serves as a~cache which maps a~given CPDBM to its unique representative in the storage.
Once a~CPDBM representative is resolved, it is saved in $M_2$. 

\begin{algorithm}[t]
  \label{algorithm:storage}
  \caption{State space storage operations}
  \SetAlgoLined
  \SetInd{0.65em}{0.65em}
%  \SetKwFor{ForEach}{foreach}{do}{endfch}
  \DontPrintSemicolon

  \SetKwProg{proc}{Procedure}{}{}

   \begin{multicols}{2}
  \setcounter{AlgoLine}{0}

  \proc{$InitializeStorage()$}
  {
    \nl $Storage \gets \emptyset$;
     $M_1 \gets \emptyset$;
     $M_2 \gets \emptyset$ \;
  }
  
  \proc{$SetData(l,C,D,data)$}
  {
    \nl \If{ $M_2(C,D) \neq \emptyset $}
    {
        \nl $(C',D') \gets M_2(C,D)$\;
        \nl $Storage(l,C',D') \gets data$\;
    }
    \nl \Else 
    {
       \nl $IH \gets integerHull(C,D)$\;
       \nl \ForEach{$(C',D')$ in $M_1(IH)$} 
       {
          \nl \If{ $\sem{C',D'} = \sem{C,D}$}
          {
             \nl $M_2(C,D) \gets (C',D')$\;
             \nl $Storage(l,C',D') \gets data$\;
          } 
       } 
       \nl $M_2(C,D) \gets (C,D)$\;
       \nl $M_1(IH) \gets M_1(IH) \cup \{(C,D)\}$\;
       \nl $Storage(l,C,D) \gets data$\;
    }
  }
    \proc{$GetData(l,C,D)$}
  {
    \nl \If{ $M_2(C,D) \neq \emptyset $}
    {
        \nl $(C',D') \gets M_2(C,D)$\;
        \nl \KwRet $Storage(l,C',D')$\;
    }
    \nl \Else 
    {
       \nl $IH \gets integerHull(C,D)$\;
       \nl \ForEach{$(C',D')$ in $M_1(IH)$} 
       {
          \nl \If{ $\sem{C',D'} = \sem{C,D}$}
          {
             \nl $M_2(C,D) \gets (C',D')$\;
             \nl \KwRet $Storage(l,C',D')$\;
          } 
       } 
       \nl $M_2(C,D) \gets (C,D)$\;
       \nl $M_1(IH) \gets M_1(IH) \cup \{(C,D)\}$\;
       \nl $Storage(l,C,D) \gets initialData$\;
       \nl \KwRet $initialData$\;
    }
  }
  \end{multicols}
\end{algorithm} 

The pseudo code of state space storage operations is given in Algorithm~2.
Note that the procedures $SetData$ and $GetData$ are analogous. 
%In our tool we have implemented these procedures using hash tables.
In our prototype tool, the two mappings as well as the storage itself are
implemented using hash tables. 
Checking whether two states are semantically equivalent is implemented using Parma Polyhedra Library~\cite{BagnaraHZ08SCP}. The library is also used to check parametric constraint satisfaction in the CPDBM operations.
%Checking of constraint satisfaction is implemented using Parma Polyhedra Library~\cite{BagnaraHZ08SCP}. 

%To refine the performance even more, there is possibility to incorporate parallel processing of successor generation, usage of unique representation of state component $C$, and even caching of state successors.

\section{Experimental evaluation}
\label{section:evaluation}
We have implemented the proposed technique for integer parameter synthesis in
our proof-of-concept tool. Our goal is to compare our method with the explicit
enumeration technique. 
To be able to compare performance of both techniques under similar conditions
we also implemented the standard DBM-based LTL model checker for timed
automata.
Both tools use the same LTL to BA translation method~\cite{gastin2001fast} and
analogous extrapolation techniques.

% model and property 
%In order to use case study example distributed with UPPAAL our implementation supports extension of timed automata modelling language. 
%This extension includes composition of timed automata into a network, synchronization channels, data in locations, urgent locations, and transition effects.
Our evaluation was performed on a~parametric extension of the case study TrainGate~\cite{behrmann2004tutorial} provided with the tool UPPAAL.
In the TrainGate model we substitute all 6 integer bounds with separate parameters and consider two trains. This model is presented in Figure~\ref{fig:traingate}.
We checked two LTL properties. 
The first property \textit{prop1} states that the two trains can not cross the
bridge simultaneously ($G!(Train_1.cross$ and $Train_2.cross)$).
The second property \textit{prop2} states that whenever the first train is
approaching the bridge it will cross the bridge eventually ($G\  Train_1.appr
\implies F\ Train_1.cross$).
%Property 1 states that train queue never overflows ($G\ Gate.queue[\#trains]==0$). 
%\textit{Prop1} is satisfied for all parameter valuations whereas \textit{prop2} is satisfied only for some parameter valuations due to possible deadlocks.
For all considered parameter valuations which do not lead to the deadlock, \textit{prop1} and \textit{prop2} are satisfied.

\begin{figure}[t]
  \centering
  \begin{minipage}{0.5\textwidth}
  \subfloat[Train]{\label{fig:train}
    \scalebox{1.0}{
      \includegraphics[scale=0.3]{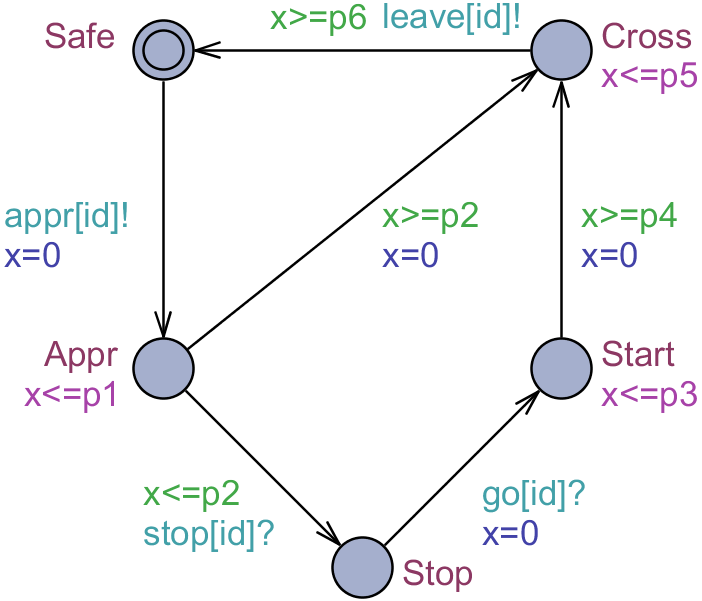}
    }
  }
  \end{minipage}
  \begin{minipage}{0.45\textwidth}
  \subfloat[Gate]{\label{fig:gate}
    \scalebox{1.0}{
    \includegraphics[scale=0.3]{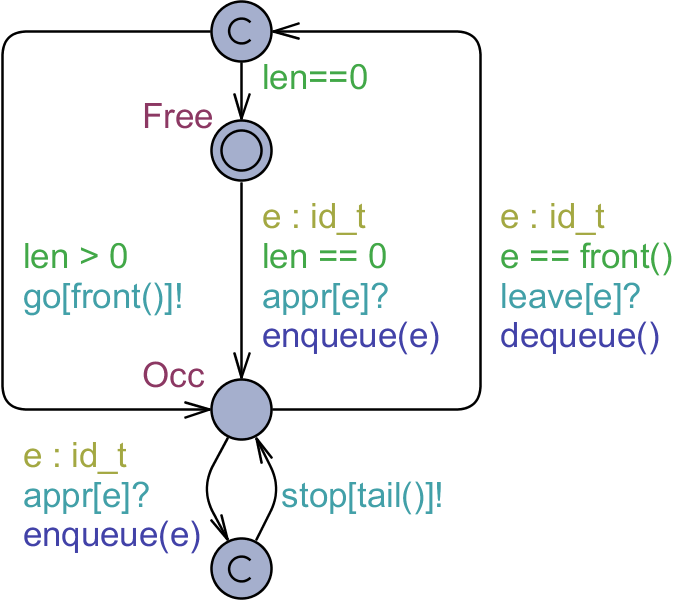}%
    }
  }
  \end{minipage}%
\caption{Parametric TrainGate Model}
\label{fig:traingate}
\end{figure}

% evaluation results
Experiments were performed on a~PC with CPU i5-4690 and 16GB RAM. We considered
a~timeout of 12 hours for each task. We provide percentage of solved parameter valuations if the timeout was reached by explicit enumeration.

 Table \ref{table:count} shows the impact of the number of parameters used in
 the model. For models
 with a~small number of parameters and small value ranges the explicit enumeration
 %and standard model checking
 can be more efficient. However, higher parameter
 count significantly favours the cumulative algorithm.
Table \ref{table:range} shows the impact of
the parameter range size on the execution times. Note that for larger parameter ranges the cumulative algorithm is faster than explicit enumeration.
%Note that the time spent by the cumulative algorithm is not proportional to the
%size of the parameter ranges. 
%We can observe that the extension of minimal
%bound has more significant influence on the time spent than the extension of maximal bound.

%Note that the time spent by the cumulative algorithm does not depend only on the
%size of the parameter ranges or parameter count.
%Table \ref{table:range} shows the impact of the ratio of parameter range sizes to values from the model guards. Performance of the explicit enumeration is not influenced by this ratio so significantly. Hovewer, cumulative algorithm can be still faster than explicit enumeration.

\begin{table}[t]
  \caption{Impact of model parameter count}
  \label{table:count}
  \centering
\begin{tabular}{| c || r | r | r | r |}
 \hline
\multirow{7}{*}{ \begin{minipage}{2.5cm}\centering TrainGate model \\ 2 trains\end{minipage}} & 
\multicolumn{1}{|c|}{3 params } &
\multicolumn{1}{|c|}{4 params } & 
\multicolumn{1}{|c|}{5 params } &
\multicolumn{1}{|c|}{6 params } \\
 & $p_1 \in [20,50]$ & $p_1 \in [20,50]$ & \multicolumn{1}{|l|}{$p_1 \in [20,50]$} & \multicolumn{1}{|l|}{$p_1 \in [20,50]$}\\
 & $p_2 \in [10,50]$ & $p_2 \in [10,50]$ & \multicolumn{1}{|l|}{$p_2 \in [10,50]$} & \multicolumn{1}{|l|}{$p_2 \in [10,50]$} \\
 &\multicolumn{1}{|l|}{$p_3 \in [15,50]$} & $p_3 \in [15,50]$ & \multicolumn{1}{|l|}{$p_3 \in [15,50]$} & \multicolumn{1}{|l|}{$p_3 \in [15,50]$} \\ 
 &\multicolumn{1}{|l|}{$p_4 = 7$} &\multicolumn{1}{|l|}{$p_4 \in [~7,50]$} &\multicolumn{1}{|l|}{$p_4 \in [~7,50]$}&\multicolumn{1}{|l|}{$p_4 \in [~7,50]$}\\ 
 &\multicolumn{1}{|l|}{$p_5 = 5$}&\multicolumn{1}{|l|}{$p_5 = 5$} &\multicolumn{1}{|l|}{$p_5 \in [~5,50]$} & \multicolumn{1}{|l|}{$p_5 \in [~5,50]$} \\ 
 &\multicolumn{1}{|l|}{$p_6 = 3$} &\multicolumn{1}{|l|}{$p_6 = 3$}&\multicolumn{1}{|l|}{$p_6 = 3$}&\multicolumn{1}{|l|}{$p_6 \in [~3,50]$}\\ 
 \hline
prop1 explicit enumeration & 0:01:03 & 0:44:50 & Timeout(51\%) & Timeout(2\%)  \\
prop1 cumulative algorithm & 0:08:16 & 0:54:39 & 3:20:25 & 7:58:42 \\
prop2 explicit enumeration & 0:01:21 & 0:58:17 & Timeout(42\%) & Timeout(1\%) \\
prop2 cumulative algorithm & 0:12:20 & 1:23:37 & 5:11:01 & 10:48:16 \\
 \hline
\end{tabular}
\end{table}

\begin{table}[t]
  \caption{Impact of parameter range size}
  \label{table:range}
  \centering
\begin{tabular}{| c || r | r | r | r | }
 \hline
\multirow{3}{*}{ \begin{minipage}{2.5cm}\centering TrainGate model \\ 2 trains \\ 4 parameters \\ $p_5=5~p_6=3$\end{minipage}} 
 &\multicolumn{1}{|l|}{$p_1 \in [20,50]$} & \multicolumn{1}{|l|}{$p_1 \in [20,100]$} & \multicolumn{1}{|l|}{$p_1 \in [10,100]$}\\
 & \multicolumn{1}{|l|}{$p_2 \in [10,50]$} & \multicolumn{1}{|l|}{$p_2 \in [10,100]$} & \multicolumn{1}{|l|}{$p_2 \in [10,100]$}\\
 & \multicolumn{1}{|l|}{$p_3 \in [15, 50]$} & \multicolumn{1}{|l|}{$p_3 \in [15, 100]$} & \multicolumn{1}{|l|}{$p_3 \in [10,100]$}\\
 &\multicolumn{1}{|l|}{$p_4 \in [ ~7, 50]$} &\multicolumn{1}{|l|}{$p_4 \in [ ~7, 100]$} & \multicolumn{1}{|l|}{$p_4 \in [10,100]$}\\ 
 \hline
prop1 explicit enumeration & 0:44:50 & Timeout(68\%) & Timeout(63\%)  \\
prop1 cumulative algorithm & 0:54:39 & 7:39:43 & 6:56:49  \\
prop2 explicit enumeration & 0:58:17 & Timeout(56\%) & Timeout(53\%) \\
prop2 cumulative algorithm & 1:23:37 & 10:25:28 & 8:59:11 \\
 \hline
\end{tabular}
\end{table}

\iffalse
\begin{table}
  \caption{Impact of parameter range size (original order)}
  \label{table:range}
  \centering
\begin{tabular}{| c || r | r | r | r | }
 \hline
\multirow{3}{*}{ \begin{minipage}{2.5cm}\centering TrainGate model \\ 2 trains \\ 4 parameters \\ $p_5=15~p_6=7$\end{minipage}} 
 &\multicolumn{1}{|l|}{$p_1 \in [20,50]$} & \multicolumn{1}{|l|}{$p_1 \in [20,100]$} & \multicolumn{1}{|l|}{$p_1 \in [10,100]$}\\
 & \multicolumn{1}{|l|}{$p_2 \in [10,50]$} & \multicolumn{1}{|l|}{$p_2 \in [10,100]$} & \multicolumn{1}{|l|}{$p_2 \in [10,100]$}\\
 & \multicolumn{1}{|l|}{$p_3 \in [5, 50]$} & \multicolumn{1}{|l|}{$p_3 \in [5, 100]$} & \multicolumn{1}{|l|}{$p_3 \in [10,100]$}\\
 &\multicolumn{1}{|l|}{$p_4 \in [ ~3, 50]$} &\multicolumn{1}{|l|}{$p_4 \in [ ~3, 100]$} & \multicolumn{1}{|l|}{$p_4 \in [10,100]$}\\ 
 \hline
prop1 explicit enumeration & 0:40:40 &  & Timeout(63\%)  \\
prop1 cumulative algorithm & 0:08:24 & 0:50:42 & 6:56:49  \\
prop2 explicit enumeration & 0:55:18 &  & Timeout(53\%) \\
prop2 cumulative algorithm & 0:14:15 &  & 8:59:11 \\
 \hline
\end{tabular}
\end{table}
\fi

% z3 note
%We have also tried to substitute calls of PPL with queries to SMT solver Z3[] where applicable. Our experimental evaluation confirmed our assumption that calls to PPL are less time consuming in our setup.

\section{Conclusion and Future Work}
\label{section:conclusion}
% general goal of the paper
We have presented an algorithmic framework for the~bounded integer parameter synthesis for parametric timed automata with an~LTL specification. The proposed framework allows the avoidance of the explicit enumeration of all possible parameter valuations. 

%Implementation
Our symbolic technique is based on the zone abstraction and uses a~parametric extension of difference bound matrices.
To be able to employ the zone-based method successfully we have introduced a~finite abstraction called the pk-extrapolation.
To be able to synthesize all violating parameter valuations we have introduced the Cumulative NDFS algorithm which is an~extension of the NDFS algorithm.

%Evaluation
We have implemented the proposed technique in an~experimental tool and our
experiments confirm that this technique can be significantly faster than the
explicit enumeration technique.

% Future fork
As for future work we plan to introduce different finite abstractions based on different extrapolations and compare their influence on the state space size.
We also plan to introduce a~parallel version of the cumulative algorithm.
% Presentation of computed counterexamples 
Other area that can be investigated is the employment of different linear specification logics, e.g.~Clock-Aware LTL~\cite{ICTAC2014} which enables the use of clock-valuation constraints as atomic propositions.

\bibliographystyle{splncs}
\bibliography{main}

%\newpage
\appendix
\input{main_appendix.tex}

\end{document}

%% file: main_appendix.tex
\spnewtheorem*{reflemma}{Lemma}{\bfseries}{\rmfamily}
\spnewtheorem*{reftheorem}{Theorem}{\bfseries}{\rmfamily}

\section{Proof of Theorem~\ref{theorem:extraLassoSimulation}}
\label{app:extraLassoSimulation}
\subsection{Finiteness of pk-extrapolation}
We start with necessary definitions. In the following, we write $s_1 \in_v S_2$
if a~concrete state $s_1$ is contained in a symbolic state $S_2$; more
precisely if $s_1=(l_1,\eta)$ is a~concrete state from $\sem{A}_v$,
$S_2=(l_2,\sem{C,D})$ is a~symbolic state from $\sem{A}$, $l_1=l_2$, $v \in C$,
and $\eta \in \sem{D}_v$.

\begin{definition}[Time-abstracting simulation]
Given an~LTS $(S, s_0,\rightarrow)$, a~time-abstracting simulation $R$ over $S$
is a~binary relation satisfying the following conditions:
\begin{itemize}
\item $s_1 R s_2$ and $s_1 \stackrel{act}{\rightarrow} s_1'$ implies the existence of $s_2 \stackrel{act}{\rightarrow} s_2'$ such that $s_1'R s_2'$, and
\item $s_1 R s_2$ and $d_1 \in \mathbb{R}^{\geq 0}$ and $s_1 \stackrel{d_1}{\rightarrow} s_1'$ implies the existence of $d_2 \in \mathbb{R}^{\geq 0}$ and $s_2 \stackrel{d_2}{\rightarrow} s_2'$ such that $s_1'R s_2'$.
\end{itemize}
We define the largest simulation relation over $S$ ($\preccurlyeq_S$) in the
following way: $s \preccurlyeq_S s' $ if there exists a~time-abstracting
simulation $R$ with $(s,s') \in R$.
When $S$ is clear from the context we shall only use $\preccurlyeq$ instead of $\preccurlyeq_S$ in the following.
\end{definition}

\begin{definition}[PTBA abstract symbolic semantics]
Let $A=(M,F)$ be a~PTBA. An \emph{abstraction over} 
$\sem{A} = (\mathbb S_A,\mathbb S_{init},\Longrightarrow)$ is a~mapping 
$\alpha: \mathbb S_A \rightarrow 2^{\mathbb S_A}$ such that the~following conditions hold:
\begin{itemize}

\item $(l',\sem{C',D'}) \in \alpha((l,\sem{C,D}))$ implies $l=l' \wedge \sem{C'} \subseteq \sem{C} \wedge \sem{C', D} \subseteq \sem{C', D'}$,

\item for each $v \in \sem{C}$ there exist $S_1,S_2$ such that $S_2=(l,\sem{C',D'}) \in \alpha(S_1)$ and for each $s \in_v S_2$ there exists a~state $s' \in_v S_1$ satisfying $s \preccurlyeq s'$.

\end{itemize}
An abstraction $\alpha$ is called \emph{finite} if its image is finite.
An abstraction $\alpha$ over $\sem{A}$ induces a~new transition system 
$\sem{A}^{\alpha}=(\mathbb Q_A,\mathbb Q_{init},\Longrightarrow^{\alpha})$ where 
\begin{itemize}
\item $\mathbb Q_A = \{ S\ |\ S \in \alpha(S')\ and\ S' \in \mathbb S_A \}$,
\item $\mathbb Q_{init} = \{ S\ |\ S \in \alpha(S')\ and\ S' \in \mathbb S_{init}\}$, and
\item $Q \Longrightarrow^{\alpha} Q'$ if there is $S \in \mathbb S_A$ such that $Q' \in \alpha(S)$ and $Q \Longrightarrow S $.
\end{itemize}
An~accepting state, a~run and an~accepting run are defined analogously as in the $\sem{A}$ case.
If $\alpha$ is finite then $\sem{A}^{\alpha}$ can be viewed as a~B\"{u}chi automaton.
\end{definition}

\begin{lemma}
\label{lemma:extraFiniteness}
Let $A$ be a~PTBA. The pk-extrapolation is a~finite abstraction over $\sem{A} = (\mathbb S_A,\mathbb S_{init},\Longrightarrow)$.
\end{lemma}
\begin{proof}
First, we prove that the pk-extrapolation is an abstraction. It is easy to see that the pk-extrapolation satisfies the first condition $(l',\sem{C',D'}) \in \alpha((l,$ $\sem{C,D}))$ implies $l=l' \wedge \sem{C'} \subseteq \sem{C} \wedge \sem{C', D} \subseteq \sem{C', D'}$. The validity of the second condition follows from the following observation. For each $v \in \sem{C}$ and each $\eta' \in \sem{D'}_v$ there exists $\eta \in \sem{D}_v$ such that for each clock~$x$ and each guard~$g$ the following implication holds: $\eta'(x) \models g \implies \eta(x) \models g$.

Now, we need to show that the pk-extrapolation is finite.
From the definition we have the fact that the number of locations is finite and the number of sets of bounded parameter valuations is finite. We need to show that there are only finitely many sets $\sem{C,D}$ when the pk-extrapolation is applied. 
This follows from the fact that for each $e_{ij}$ from $D$ and $v \in C$ the
expression $\sem{e_{ij}}_v$ can be evaluated only to a~value from the finite set
$\{ -M(x_i),-M(x_i)+1,\ldots, M(x_i)-1, M(x_i) ,\infty \}$.

\qed
\end{proof}

\subsection{Preservation of accepting runs}
We transform the proof of Theorem 1 of \cite{li09formal} and all corresponding lemmata into our parametric setup. For the sake of simplicity of the proof, we add labels to the transitions in $\sem{A}^{\alpha}$ in the following way. For each transition we use the location of a~source state as the transition label. Since labels are not used in the proposed method, it is safe to do that.

For the rest, let $v$ be a~parameter valuation, $A$ be a~PTBA, and $\alpha$ be a finite abstraction over $\sem{A}$. Then, we denote by $A\mid v$ a~timed B\"{u}chi automaton obtained from $A$ by replacing each parameter $p$ with the value $v(p)$. 
We use~$\equiv_{N}$ to denote the standard region abstraction~\cite{li09formal} over the timed automaton~$N$. 

We write $s_1 \cactiond{act_1,act_2,\ldots,\act_{k-1}} s_{k}$ if there exist $s_2,\ldots,s_{k-1}$ such that $s_1 \cactiond{act_1} s_2$, $s_2 \cactiond{act_2} s_3$, $\ldots$, and $s_{k-1} \cactiond{act_{k-1}}s_{k}$.
We write $s\cactiond{act_1,act_2,\ldots,\act_{k}}^{*} s'$ if $s \cactiond{act_1,act_2,\ldots,\act_{k}} s'$ or there exist some $s_1,s_2,\ldots,s_n (n \geq 1)$ such that \\$s \cactiond{act_1,act_2,\ldots,\act_{k}} s_1$,$s_1 \cactiond{act_1,act_2,\ldots,\act_{k}} s_2$,$\ldots$, $s_{n-1} \cactiond{act_1,act_2,\ldots,\act_{k}}s_{n}$, and $s_{n} \cactiond{act_1,act_2,\ldots,\act_{k}}s'$.

% li lemma 1
\begin{lemma}
\label{lemma:li1}\textbf{\upshape\cite{li09formal}}
The equivalence relation $\equiv_{A\mid v}$ is a~time-abstracting bisimulation.
\end{lemma}

% li lemma 2
\begin{lemma}
\label{lemma:li2}\textbf{\upshape\cite{li09formal}}
Let $s_1$,$s_1'$,$s_2$ be concrete states in $\sem{A}_v$, $R$ be a time-abstracting simulation and $s_1Rs_1'$. If $s_1 \cactiond{act_1,act_2\ldots,act_k} s_2$, then there exists a~concrete state $s_2'$ in $\sem{A}_v$ such that $s'_1 \cactiond{act_1,act_2\ldots,act_k} s'_2$.
\end{lemma}

% li lemma 3
\begin{lemma}
\label{lemma:li3}
Let $s_1$,$s_2$ be concrete states in $\sem{A}_v$. If $s_1 \cactiond{act_1,act_2\ldots,act_k} s_2$ and $s_1 \equiv_{A\mid v} s_2$, then there is an infinite sequence of concrete states $s_1 s_2 \ldots $ in $\sem{A}_v$ such that for each $i \geq 1$, $s_i \cactiond{act_1,act_2,\ldots, \act_k} s_{i+1}$.
\end{lemma}
\begin{proof}
We define $s_k$ ($k=3,4,\ldots$) by induction on $k$.

\textbf{Basis:} By Lemma \ref{lemma:li1} and lemma \ref{lemma:li2} there exists $s_3$ such that 
$s_2 \cactiond{act_1,act_2,\ldots, \act_k} s_{3}$, and $s_2 \equiv_{A\mid v } s_3$.

\textbf{Assumption:} Assume that we have $s_1,s_2,\dots,s_k$ such that for each $i \in \{1,2,\ldots, k-1\}$, 
$s_i \cactiond{act_1,act_2,\ldots, \act_k} s_{i+1}$, and $s_{k-1} \equiv_{A\mid v } s_k$.

\textbf{Step:} From $s_{k-1} \cactiond{act_1,act_2,\ldots, \act_k} s_{k}$, and $s_{k-1} \equiv_{A\mid v } s_k$, by Lemma \ref{lemma:li2} there exists $s_{k+1}$ such that $s_{k} \cactiond{act_1,act_2,\ldots, \act_k} s_{k+1}$, and $s_{k} \equiv_{A\mid v } s_{k+1}$.

Thus, using induction, we get an infinite sequence of states $s_1s_2\ldots$ such that for each $i \geq 1 $, 
$s_i \cactiond{act_1,act_2,\ldots, \act_k} s_{i+1}$.
\qed
\end{proof}

% li lemma 5
\begin{lemma}
\label{lemma:li5}
\label{lemma:simulation}
Let $s',s_1,s_2$ be concrete states in $\sem{A}_v$ and $S_1$,$S_2$ be symbolic states in $\sem{A}$.
\begin{enumerate}
\item If $S_1 \Longrightarrow S_2$ and $s' \in_v S_2$, then there exist concrete state $s$ in $\sem{A}_v$ such that $ s \stackrel{act}{\longrightarrow}_d s'$.
\item If $s_1 \stackrel{act}{\longrightarrow}_d s_2$ and $s_1 \in_v S_1$, then $S_1 \Longrightarrow S_2$ for some symbolic state $S_2$ in $\sem{A}$ with $s_2 \in_v S_2$.
\end{enumerate}
\end{lemma}
\begin{proof}
We refer the reader to the proofs of Lemma 3.16 and Lemma 3.18 in \cite{hune02linear}.
\end{proof}

% li lemma 6
\begin{lemma}
\label{lemma:li6}
\label{lemma:simulationWithAbstraction}
Let $s,s_1$,$s_2$ be concrete states in $\sem{A}_v$, and $Q_1$,$Q_2$ be symbolic states in $\sem{A}^{\alpha}$.
\begin{enumerate}
\item If $Q_1 \Longrightarrow_{\alpha} Q_2$ and $s \in_v Q_2$, then there exist concrete states $s'_1,s'_2$ in $\sem{A}_v$ such that $ s'_1 \stackrel{act}{\longrightarrow}_d s'_2, s'_1 \in_v Q_1$ and $s \preccurlyeq s'_2$.
\item If $s_1 \stackrel{act}{\longrightarrow}_d s_2$ and $s_1 \in_v Q_1$, then $Q_1 \Longrightarrow_{\alpha} Q_2$ for some symbolic state $Q_2$ in $\sem{A}^{\alpha}$ with $s_2 \in_v Q_2$.
\end{enumerate}
\end{lemma}
\begin{proof}
\begin{enumerate}
\item
From $Q_1 \Longrightarrow_{\alpha} Q_2$ we know that there exists $S$ such that $Q_1 \Longrightarrow S$ and $Q_2 \in \alpha(S)$. For any $s \in_v Q_2$, since $Q_2 \in \alpha(S)$, there is $s'_2 \in_v S$ such that $s \preccurlyeq s'_2$. Since $Q_1 \Longrightarrow S$ and $s'_1 \in_v S$, by Lemma \ref{lemma:simulation}, there is a~$s'_1 \in_v Q_1$ such that $ s'_1 \stackrel{act}{\longrightarrow}_d s'_2$.
\item
By Lemma \ref{lemma:simulation} there is a~$S$ such that $Q_1 \Longrightarrow S$ with $s_2 \in_v S$. Let $Q_2 \in \alpha(S)$ then $Q_1 \Longrightarrow_{\alpha} Q_2$ and $s_2 \in_v Q_2$.
\end{enumerate}
\end{proof}

% li lemma 7
\begin{lemma}
\label{lemma:li7}
Let $s$ be a~concrete state in $\sem{A}_v$, $Q_1$,$Q_2$ be symbolic states in $\sem{A}^{\alpha}$.
If $Q_1 \saction{act_1,act_2,\ldots,act_k} Q_2$, 
and $s\in_v Q_2$, then there exist concrete states $s_1,s_2$ in $\sem{A}_v$ such that $s_1 \in_v Q_1$, $s_1 \cactiond{act_1,act_2,\ldots,act_k} s_2$,
and $s \preccurlyeq s_2$. 
\end{lemma}
\begin{proof}
We prove the lemma by induction on $k$.

\textbf{Basis:} 
By Lemma \ref{lemma:li6}, the lemma is true for $k=1$.

\textbf{Assumption:} 
Assume that lemma holds for $k=n$.

\textbf{Step:} 
Now we prove the lemma for $k=n+1$. $Q_1 \saction{act_1,act_2,\ldots,act_{n+1}}_{\alpha} Q_2$ implies that there exists  a $Q \in \sem{A}^{\alpha}$ such that  $Q_1 \saction{act_1,act_2,\ldots,act_n}_{\alpha} Q$ and $Q \saction{act_{n+1}}_{\alpha} Q_2$.
 By Lemma \ref{lemma:li6} and the fact that $Q \saction{act_{n+1}}_{\alpha} Q_2$ and $s\in_{v} Q_2$, we have $s'$ and $s''$ such that $s' \in_{v} Q$, $s' \cactiond{act_{n+1}} s''$, and $s \preccurlyeq s''$.
 Since $Q_1 \saction{act_1,act_2,\ldots,act_n}_{\alpha} Q$ and $s' \in_{v} Q$,
 by the induction assumption there exist $s_1$ and $s'''$ such that $s_1
 \in_{v} Q_1$, $s_1 \cactiond{act_1,act_2,\ldots,act_n} s'''$ ,and $s'
 \preccurlyeq s'''$.
 Since $ s' \preccurlyeq s''' $ and $ s' \cactiond{act_{n+1}} s''$, by Lemma \ref{lemma:li2} it follows that there is a $s_2$ such that $s''' \cactiond{act_{n+1}} s_2$ and $s'' \preccurlyeq s_2$. From the fact that $\preccurlyeq$ is transitive and $s\preccurlyeq s'' $ and  $ s'' \preccurlyeq s_2$ we have $s \preccurlyeq s_2$.
 
 By $ s_1 \cactiond{act_1,act_2,\ldots,act_n} s''' $ and $ s''' \cactiond{act_{n+1}} s_2 $ we obtain $ s_1 \cactiond{act_1,act_2,\ldots,act_n+1} s_2 $.
\qed
\end{proof}

% li lemma 8
\begin{lemma}
\label{lemma:li8}
Let $s$ be a~concrete state in $\sem{A}_v$, $Q_1$,$Q_2$ be symbolic states in $\sem{A}^{\alpha}$.
If $Q \sactiona{act_1,act_2,\ldots,act_k} Q$ and $s \in_v Q$, then for any $n \geq 1$, there exist concrete states $s_1,s_2 \ldots s_{n+1}$ in $\sem{A}_v$ such that $s_1 \in_v Q$, $s \preccurlyeq s_{n+1}$, and for each $i \in {1,2,\ldots,n}$, $s_i \cactiond{act_1,act_2,\ldots,act_k} s_{i+1}$.
\end{lemma}
\begin{proof}
We prove the lemma by induction on $n$.

\textbf{Basis:} 
By Lemma \ref{lemma:li7}, the lemma is true for $n=1$.

\textbf{Assumption:} 
Assume that lemma holds for $n=m$.

\textbf{Step:} 
Now we prove that the lemma is true for $n=m+1$.
Since $Q \sactiona{act_1,act_2,\ldots,act_k} Q$ and $s \in_{v} Q$, by Lemma \ref{lemma:li7}, there exist $s',s''$ such that $s' \in_{v} Q$, $s' \cactiond{act_1,act_2,\ldots,act_k} s''$, and $ s \preccurlyeq s''$. Applying the induction assumption to $Q \sactiona{act_1,act_2,\ldots,act_k} Q$ and $s' \in_{v} Q$, we know that there exist $s_1,s_2 \ldots s_{m+1}$ such that $s_1 \in_{v} Q$, $s' \preccurlyeq s_{m+1}$, and for each 
$i \in {1,2,\ldots,m}$, $s_i \cactiond{act_1,act_2,\ldots,act_k} s_{i+1}$.

Since $s'\cactiond{act_1,act_2,\ldots,act_k} s''$ and $s' \preccurlyeq s_{m+1}$, by Lemma \ref{lemma:li2}, there exists  $s_{m+2}$ such that $ s_{m+1}\cactiond{act_1,act_2,\ldots,act_k} s_{m+2}$, and $s'' \preccurlyeq s_{m+2}$.

Since $s \preccurlyeq s''$ and $s'' \preccurlyeq s_{m+2}$ we obtain $s \preccurlyeq s_{m+2}$, thus the lemma holds for $n=m+1$.
\qed
\end{proof}

% li lemma 9
\begin{lemma}
\label{lemma:li9}
Let $s$ be a~concrete state in $\sem{A}_v$, $Q_1$,$Q_2$ be symbolic states in $\sem{A}^{\alpha}$.
If $Q \sactiona{act_1,act_2,\ldots,act_k} Q$ and $s \in_v Q$, then there exist concrete states $s_1,s_2 \ldots s_{m}$ in $\sem{A}_v$ and $i \in \{1,2,\ldots,m-1\}$ such that $s_1 \in_v Q$ , $s_i \equiv_{A\mid v} s_m$, and for each $j \in \{ 1,2,\ldots,m-1\}$, $s_j \cactiond{act_1,act_2,\ldots,act_k}s_{j+1}$.
\end{lemma}
\begin{proof}
We know that there are only finitely many $\equiv_{A\mid v}$-equivalence classes. Let $n$ be an integer greater than the number of $\equiv_{A\mid v}$-equivalence classes. By Lemma \ref{lemma:li8}, there exist $s_1,s_2 \ldots s_{n+1}$ such that $s_1 \in_{v} Q$, and for each $j \in \{ 1,2,\ldots,n\}$, $s_j \cactiond{act_1,act_2,\ldots,act_k}s_{j+1}$.

Since the sequence of states $s_2,s_3 \ldots s_{n+1}$ has length $n$, there exist $i,m \in \{ 2,3,\ldots,n+1\}$ such that $i < m$ and $s_i \equiv_{A\mid v} s_m$.
\qed
\end{proof}

% li lemma 10
\begin{lemma}
\label{lemma:li10}
Let $Q=(l,\sem{C,D})$ be symbolic states in $\sem{A}^{\alpha}$ such that $v \in C$.
If $Q \sactiona{act_1,act_2,\ldots,act_k} Q$, then there exist concrete states $s',s'',s'''$ in $\sem{A}_v$ such that $s' \in_v Q$, $s' \cactiond{act_1,act_2,\ldots,\act_k}^{*} s''$, $s'' \cactiond{act_1,act_2,\ldots,act_k}^{*} s'''$ and $s'' \equiv_{A\mid v } s'''$.
\end{lemma}
\begin{proof}
Since $Q \sactiona{act_1,act_2,\ldots,act_k} Q$, by definition, there exist $s \in_{v} Q$. By Lemma \ref{lemma:li9}, there exist $s_1,s_2,s_3,\ldots,s_m$ and $ i \in \{ 2,3,\ldots, m-1\}$ such that $s_1 \in_{v} Q$, $s_i \equiv_{A\mid v } s_m$, and for each $j \in \{ 1,2,\ldots,m-1\}$, $s_j \cactiond{act_1,act_2,\ldots,act_k}s_{j+1}$.

Let $s'=s_1$,$s''=s_i$, and $s'''=s_m$, then $s' \in_{v} Q$, $s' \cactiond{act_1,act_2,\ldots,act_k}^{*} s''$, $s'' \cactiond{act_1,act_2,\ldots,act_k}^{*} s'''$, and $s'' \equiv_{A\mid v } s'''$.
\qed
\end{proof}

% li lemma 11
\begin{lemma}
\label{lemma:li11}
Let $s_1,s_2$ be concrete states in $\sem{A}_v$.
If $s_1 \cactiond{act_1,act_2,\ldots,act_k}^{*} s_2$ and $s_1 \equiv_{A\mid v} s_2$, then there is an infinite sequence of concrete states $s_1s_2\ldots$ in $\sem{A}_v$ such that for each $i \geq 1$, $s_i \cactiond{act_1,act_2,\ldots,act_k} s_{i+1}$.
\end{lemma}
\begin{proof}
Follows from  Lemma \ref{lemma:li3}.
\qed
\end{proof}

%lemma 12 Li
\begin{lemma}
\label{lemma:li12}
Let $Q=(l,\sem{C,D})$ be a~symbolic state in $\sem{A}^{\alpha}$ such that $v \in C$.
If $Q \sactiona{act_1,act_2,\ldots,act_k} Q$, then there is an infinite sequence of concrete states $s_1s_2\ldots$ in $\sem{A}_v$ such that $s_1 \in_v Q$, and for each $i \geq 1$, $s_i \cactiond{act_1,act_2,\ldots,act_k}^{*} s_{i+1}$.
\end{lemma}
\begin{proof}
By Lemma \ref{lemma:li10}, there exist $s_1,s_2,s_3$ such that $s_1 \in_{v} Q$, $s_1 \cactiond{act_1,act_2,\ldots,act_k}^{*} s_2$, $s_2 \cactiond{act_1,act_2,\ldots,act_k}^{*} s_{3}$ and $s_2 \equiv_{A\mid v} s_3$.

By Lemma \ref{lemma:li11}, there is an infinite sequence of states $s_2,s_3,s_4,\ldots$ such that for each $i \geq 2$, $s_i \cactiond{act_1,act_2,\ldots,act_k}^{*} s_{i+1}$.%
\qed
\end{proof}

% main li theorem
\begin{theorem}
\label{theorem:lassoSimulation}
Let $A=((L, l_0, X, P, \Delta, \Inv), F)$ be a~PTBA and $\alpha$ be a~finite abstraction. For each parameter valuation $v$ the following holds: there exists an accepting run of
$\sem{A}_v$ if and only if there exists an accepting run respecting $v$ of $\sem{A}^{\alpha}$.
\end{theorem}
\begin{proof}
The fact that the existence of an accepting run of $\sem{A}_v$ implies the existence of an~accepting run respecting $v$ of $\sem{A}^{\alpha}$ can be proved easily for each valuation $v$ by induction and Lemma~\ref{lemma:li6}.

Now we give the proof for the other direction. If $\sem{A}^{\alpha}=$
$(\mathbb Q_A, \mathbb Q_{init},$ $\Longrightarrow^{\alpha})$ over $L$ has an
accepting run respecting $v$, then there exists a $Q=(l,\sem{C,D}) \in \mathbb Q_A$ and
$act_0,act_1,\ldots,act_i,\ldots,act_k \in L$ such that $Q_0
\sactiona{act_0,act_1,\ldots,\act_{i-1}} Q$, and $Q
\sactiona{act_i,act_{i+1},\ldots,act_k} Q$, and $F \cap \{act_i,$
$act_{i+1},\ldots,$ $act_k\} \neq \emptyset$ where $Q_0$ is the initial state
of $\sem{A}^{\alpha}$ and $v \in C$.

Applying Lemma~\ref{lemma:li12} to $Q_i \sactiona{act_i,act_{i+1},\ldots,act_k} Q_{i}$ we have  an infinite sequence of states $s_2's_3's_4'\ldots$ such that $s_2' \in_v Q_i$ and for each $j \geq 2$ \\ $s_j' \cactiond{act_i,act_{i+1} \ldots, act_k}^{*} s_{j+1}'$.

Applying  Lemma~\ref{lemma:li7} to $Q_0 \sactiona{act_0,act_1,\ldots,\act_i-1} Q_i$ and $s_2' \in_v Q$, it follows that there exist $s',s''$ such that $s' \in_v  Q_0$, $s' \cactiond{act_0,act_1,\ldots,act_{i-1}} s''$, and $s_2' \preccurlyeq s''$.

By $s' \in_v Q_0$ and $Q_0 \in \alpha(S_0)$, we know that there exists a $s_1 \in_v S_0$ such that $s' \preccurlyeq s_1$.
From the fact that $s' \cactiond{act_0,act_1,\ldots,act_{i-1}}s''$, and $s' \preccurlyeq s_1$, we know that there exists a $s_2$ such that $s_1 \cactiond{act_0,act_1,\ldots,act_{i-1}} s_2$, and $s'' \preccurlyeq s_2$. Thus we have obtained that $s_2' \preccurlyeq s_2$.

Applying  Lemma~\ref{lemma:li2} to $s_2' \preccurlyeq s_2$ and $s'_j \cactiond{act_i,act_{i+1}\ldots,act_k}^{*} s'_{j+1} (j = 2,3,\ldots)$, we can obtain an infinite sequence of states $s_3s_4\ldots$ such that $s_2 \cactiond{act_i,act_{i+1}\ldots,act_k}^{*}s_3 \cactiond{act_i,act_{i+1}\ldots,act_k}^{*} s_4\ldots $.

Furthermore, from the fact that $s_1 \in_v S_0$ it follows that  there is a~$d \in \mathbb{R}^{\geq 0}$ such that $s_0 \xrightarrow{d} s_1$ where $s_0$ is the initial state of $\sem{A}_v$.

Thus, we have proved that there exists an infinite sequence of states $s_1,s_2,\ldots$ such that $s_0 \delayd 
s_1 \cactiond{act_0,act_1,\ldots,act_{i-1}} 
s_2 \cactiond{act_i,act_{i+1},\ldots,act_{k}}^{*} 
s_3 \cactiond{act_i,act_{i+1},\ldots,act_{k}}^{*} 
s_4\ldots$. 
Now, by the fact that $F \cap \{act_i,act_{i+1},\ldots,act_k\} \neq \emptyset$, 
we know that $\sem{A}_v$ has an infinite accepting run.
\qed
\end{proof}

Finally, we provide the proof of Theorem~\ref{theorem:extraLassoSimulation}.

\begin{reftheorem}[Theorem~\ref{theorem:extraLassoSimulation}]
Let $A$ be a PTBA. The pk-extrapolation is a~finite abstraction that preserves all accepting runs of $\sem{A}_v$ for each parameter valuation $v$.
%Let $A$ be a~PTBA and $\alpha$ be a~finite abstraction. For each parameter valuation $v$ the following holds: there exists an accepting run of $\sem{A}_v$ if and only if there exists an accepting run respecting $v$ of $\sem{A}^{\alpha}$.
\end{reftheorem}

\begin{proof}
Follows directly from Lemma~\ref{lemma:extraFiniteness} and Theorem~\ref{theorem:lassoSimulation}.
\end{proof}

\section{Proof of Theorem~\ref{theorem:CumulativeNDFSCorrectness}}
\label{app:cumulative}

\begin{lemma}
\label{lemma:setAccepted}
If the valuation $v$ is added to the set $\mathit{Found}$ then $v$ is returned
by the algorithm in the set  $Accepted$.
\end{lemma}
\begin{proof}
This follows from the fact that the set \textit{Found} is never decreased and at the end of computation it is assigned to \textit{Accepted}.
\end{proof}

\begin{lemma}
\label{lemma:extendedPostorder}
Let $A$ be a PTBA and  $q$  be a state in $\sem{A}$ that does not appear on any cycle under $v$. 
The \textit{OuterDFS} procedure will backtrack from q only after every reachable state $s$ such that $v \in  s.\sem{C}$ is already backtracked or $s.\sem{C} \subseteq \mathit{Found}$.
\end{lemma}
\begin{proof}
Consider an~arbitrary state $s$ such that $s$ is reachable from $q$. 
At the time of backtracking from $q$ there are two cases:
\begin{itemize}
\item
Every path from $q$ to the state $s$ contains a~state $s'$ such that $  s'.\sem{C} \subseteq  \mathit{Found}$. 
The fact that $s$ is reachable from $s'$ implies $ s.\sem{C} \subseteq s'.\sem{C}$ (using Lemma~\ref{lemma:mono}).
Hence,  $ s.\sem{C} \subseteq  \mathit{Found}$.
\item There exists a~path from $q$ to the state $s$ such that for every state $s'$ on that path it holds that $s'.\sem{C} \not\subseteq  \mathit{Found}$. 
In this case, the \textit{OuterDFS} procedure has visited state $s$ with state $q$ on the stack.
 Hence, the \textit{OuterDFS} procedure backtracks from the state $q$ after backtracking from $s$.
\end{itemize}
\end{proof}

\begin{lemma} 
\label{lemma:partialCorrectness}
For every parameter valuation $v$, the \textit{Cumulative NDFS} algorithm
returns the set \textit{Accepted} containing the valuation $v$ if and only if
the given graph contains an~accepting cycle $c$ under the valuation $v$. 
\end{lemma}
\begin{proof}
Whenever the algorithm returns a set \textit{Accepted} containing $v$ there exists an~accepting cycle $c$ under $v$. 
Such an accepting cycle can be constructed using \textit{OuterDFS} and \textit{InnerDFS} search stack at the time of adding the valuation $v$ to the set \textit{Found}.

The difficult case is to show that whenever there exists an accepting cycle under $v$ in the given graph then the algorithm returns a~set \textit{Accepted} containing~$v$. 
Suppose an~accepting cycle under a~valuation $v$ exists in the given graph and the algorithm returns a set \textit{Accepted} such that $v \not\in \mathit{Accepted}$. 
  
Let $s_i$ be an initial state in the given graph. Notice that for each state $s$ such that $v \not \in s.\sem{C} $ it holds that if $s'$ is an ancestor of state $s$ then $v \not\in  s'.\sem{C}$ (using Lemma \ref{lemma:mono}).
Hence, using the assumption $v \not\in \mathit{Accepted}$  we get that the \textit{OuterDFS} procedure visits each state $s$ such that $v \not\in  s.\sem{C}$.

Let $q$ be the first accepting state on a~cycle under $v$ from which \textit{InnerDFS} is started. 
There are two cases:
\begin{itemize}
\item There exists a~path from a~state $q$ to some state on the stack of \textit{OuterDFS} and each state $s$ on the path is unvisited by \textit{InnerDFS} and $s.\sem{C} \not \subseteq \mathit{Found}$ at the time of starting  \textit{InnerDFS} from $q$.
\item For all paths from a~state $q$ to some state $p$ on the stack of \textit{OuterDFS} there exists a~state $s$ on the path such that $ s.\sem{C}  \subseteq \mathit{Found}$ or $s$ is a state already visited by \textit{InnerDFS}.
\end{itemize}

For the first case the algorithm will detect an accepting cycle as expected and will add the valuation $v \in q.\sem{C}$ to the set \textit{Found}.
 From Lemma \ref{lemma:setAccepted} we get $v \in \mathit{Accepted}$ and we have reached a~contradiction with the assumption $v \not\in \mathit{Accepted}$.

For the second case,
whenever the path $q \leadsto p$ contains a~state $s$ such that $v \in  s.\sem{C}  \subseteq \mathit{Found}$ we reach contradiction (using Lemma \ref{lemma:setAccepted}).
Assume that for each state $s$ on  path $q \leadsto p$ it holds that $v \not \in s.\sem{C}$. 
Let $r$ be the first visited state that is reached from $q$ during \textit{InnerDFS} and is on a~cycle through $q$.  
Let $q'$ be an~accepting state that started \textit{InnerDFS} in which $r$ was visited for the first time. 
Notice the fact that \textit{InnerDFS} was started from $q'$ before starting from $q$.
There are two cases:
\begin{itemize}
\item The state $q'$ is reachable from $q$. 
Then there is an~accepting cycle $c' = q' \leadsto r \leadsto q \leadsto q'$.  
If $c'$ contains a~state $s$ such that $v \in s.\sem{C}  \subseteq \mathit{Found}$ we reach a~contradiction using Lemma \ref{lemma:setAccepted}. 
Suppose there is no state $s$ with $v \in  s.\sem{C}  \subseteq \mathit{Found}$ on the cycle $c'$. 
The cycle $c'$ was not found previously. 
However, this contradicts our assumption that $q$ is the first accepting state from which we missed a cycle.

\item The state $q'$ is not reachable from $q$. 
Notice the fact that $v \in q'.\sem{C}$ (this follows from Lemma
\ref{lemma:mono}) and therefore every cycle containing the state $q'$ is a~cycle
under $v$. %(using lemma monotonocity).
 If $q'$ appears on a cycle, then an~accepting cycle under $v$ was missed before starting \textit{InnerDFS} from $q$, contrary to our assumption.
If $q'$ does not apper on a cycle then by Lemma \ref{lemma:extendedPostorder} we backtracked from $q$ in the \textit{OuterDFS} before backtracking from $q'$ and therefore \textit{InnerDFS} started from $q$ before starting from $q'$. 
We have reached a~contradiction with the fact that  \textit{InnerDFS} started from $q'$ before  starting from $q$.
\end{itemize}
\end{proof}

\begin{lemma}
\label{lemma:termination}
The CumulativeNDFS algorithm always terminates.
\end{lemma}
\begin{proof}
From the fact that the number of vertices is finite we get that the size of the sets \textit{Inner} and \textit{Outer} is bouned.
Each invocation of \textit{InnerDFS} (\textit{OuterDFS}) procedure increases the size of the set \textit{Inner} (\textit{Outer}). 
Hence, the \textit{CumulativeNDFS} algorithm cannot proceed infinitely due to the upper bound on the size of the set \textit{Inner} and \textit{Outer}.
\end{proof}

\begin{reftheorem}[Theorem~\ref{theorem:CumulativeNDFSCorrectness}] 
Let $A$ be a~PTBA and $\alpha$ an~abstraction over $\sem{A}$. 
A~parameter valuation $v$ is contained in the output of the \textit{CumulativeNDFS($\sem{A}^{\alpha}$)} if and only if 
there exists an~accepting run respecting $v$ of $\sem{A}^{\alpha}$.
\end{reftheorem}
\begin{proof}
By Lemma \ref{lemma:termination} the algorithm is guaranteed to terminate returning the set \textit{Accepted}. 
The partial correctness, the $\Rightarrow$ case: By Lemma \ref{lemma:partialCorrectness} for each $v \in \mathit{Accepted}$ there exists an~accepting cycle under $v$ and for each $v \not\in \mathit{Accepted}$ there is no accepting cycle under $v$.
The partial correctness, the $\Leftarrow$ case: Analogously.
\end{proof}